\documentclass[12pt,fleqn,leqno]{article}
\usepackage{amsmath,amssymb}
\usepackage{amsthm,natbib,color}
\usepackage[dvips]{graphicx}

\usepackage{setspace}

\allowdisplaybreaks

\numberwithin{equation}{section}

\makeatletter
\renewcommand{\section}{\@startsection{section}{1}{0pt}{20pt}{6pt}{\large\bf}}
\renewcommand{\@seccntformat}[1]{\csname the#1\endcsname.\ }

\def\footnoterule{\kern -3pt \hrule width 2.7 true cm \kern 2.6pt}

\def\ni{\noindent}
\def\vs{\vspace}
\def\hs{\hspace}

\def\EE{\mathsf E}
\def\PP{\mathsf P}
\def\cC{\mathcal C}
\def\cD{\mathcal D}
\def\cK{\mathcal K}
\def\cF{{\cal F}}
\def\R{I\!\!R}
\def\L{I\!\!L}
\def\wt{\widetilde}
\def\wh{\widehat}

\def\e{\text{e}}

\newcommand{\eps}{\varepsilon}
\newcommand{\p}{\! +\! }
\newcommand{\m}{\! -\! }

\newtheorem{theorem}{Theorem}[section]

\newtheorem{remark}[theorem]{Remark}

\begin{document}

\title{\textbf{Optimal Mean-Reverting  Spread Trading: \\Nonlinear Integral Equation Approach  }}
\author{Yerkin Kitapbayev\thanks{Questrom School of Business, Boston University, Boston MA 02215; email:\,\mbox{yerkin@bu.edu}.
} \and Tim Leung\thanks{Applied Math Dept,  University of Washington, Seattle WA 98195; email:\,\mbox{timleung@uw.edu}.} }
\date{\today }
\maketitle


{\par \leftskip=2.6cm \rightskip=2.6cm \footnotesize
 We study several optimal stopping problems that arise from trading a mean-reverting price spread over a finite horizon.  Modeling the    spread by the Ornstein-Uhlenbeck  process, we analyze three different trading strategies: (i) the long-short strategy; (ii) the short-long strategy, and (iii) the chooser strategy, i.e. the trader can enter into the   spread by taking either long or short position.
In each of these   cases, we solve  
an  optimal double stopping problem to determine the   optimal timing for starting and subsequently closing the position.   We utilize  the local time-space calculus of  \cite{Peskir2005a} and derive the nonlinear integral equations of Volterra-type that uniquely characterize  the  boundaries associated with the optimal timing decisions in all three problems.  These  integral equations are   used to numerically compute the optimal  boundaries. 
   \par}

\vspace{10pt}

\begin{small}



{{Keywords:} spread trading, optimal stopping,
OU process, free-boundary problem,\\
local time-space calculus, integral equation}
\end{small}

\vs{-18pt}

\vs{-18pt}


\newpage
\section{Introduction}
Spread trading is a common strategy used by many traders in various markets, including equity,  fixed income, currency, and futures markets. In a spread trade, traders construct  a mean-reverting spread  by simultaneously taking positions in two or more highly correlated or co-moving assets. The proliferation of exchange-traded funds (ETFs) has further popularized spread trading  
as some ETFs are designed to replicate identical or similar assets/index. The strategy involves    opening and subsequently closing a position based on the sign and magnitude of the spread.   Therefore,  the risk of trading changes  from  the directional price movements of each asset to  the fluctuation of the spread over time. 



The core of the spread trading    strategies lies in the timing to enter and exit the market. For example, \cite{gatev2006pairs} examined  the
historical returns from the buy-low-sell-high strategy where the entry and exit  levels
are set as $\pm 1$ standard deviation from the long-run mean. Similarly,  \cite{Avellaneda2010}
 considered starting and ending a pairs trade based on a fixed distance of the spread from its mean. In \cite{elliott2005pairs}, the market entry timing is modelled by the
first passage time of an Ornstein-Uhlenbeck (OU) process, followed by an exit on a fixed future date.
In these studies, the   trading rules  are not derived endogenously based  on a given objective but are  prescribed in an \textit{ad hoc}  manner. While these naive trading rules have the advantage of being very simple and explicit, their common drawback is the lack of optimality justification. 

Alternatively,  a host of related studies apply  stochastic  optimal control techniques to determine the  optimal timing strategies for an OU price spread.  \cite{Ekstrom2010}  analyzed a   optimal single stopping problem for liquidating   a spread  position on an  infinite horizon under the OU model with no transaction
costs. \cite{song2013optimal} considered  an optimal switching approach for trading a spread repeated   over an infinite horizon with transaction costs, and solved for the optimal entry and exit thresholds.  Also over an infinite horizon, \cite{LeungLi2014OU,meanreversionbook2016},   solved anlaytically and numerically an  optimal double stopping problem to obtain the entry and exit levels  for trading an OU price spread with transaction costs as well as    a stop-loss exit. Also,   \cite{song2009stochastic} proposed and implemented a numerical stochastic approximation scheme to
solve for the optimal buy-low-sell-high strategies over a finite horizon.


 In this paper, we study several optimal stopping problems that arise from trading a mean-reverting   price spread over a finite horizon.  We model the stochastic spread directly by an OU process, and  analyze three different trading strategies. The first one involves starting by going long on the spread and reverse the position to close, and the second   represents the opposite sequence of trades -- short to open, long to close. Moreover, as  the trader ponders when to enter the market, he/she can enter by starting either with a long or short position. This gives the trader a \textit{chooser option} to be exercised upon market entry. Once the first position is committed, the trader faces an optimal timing problem to exit the market.  For each of these strategies, we solve  
an optimal double stopping problem to determine the   optimal timing for starting and subsequently closing the position.

Our method of solution utilizes    local time-space calculus of  \cite{Peskir2005a} to derive the Volterra-type integral equations  that uniquely characterize  the  boundaries associated with the optimal timing decisions in all three trading problems.  The  nonlinear equations are   useful not only for the analytical representation but also numerical computation of the value functions and optimal boundaries.   Unlike its perpetual analogue (see e.g. \cite{HFTbook}, and \cite{LeungLi2014OU}), the finite-horizon trading problems studied herein do not admit  closed-form expressions for the value functions or optimal boundaries. Nevertheless, our application of    local time-space calculus allows us to express the optimal enter and exit boundaries as unique solutions to recursive integral equations. We provide illustrations of the boundaries and discuss their properties and financial implications. 


Our paper contributes to the literature of  optimal spread trading by providing an optimal double stopping  approach together with Volterra-type  integral equations for determining and analyzing the optimal boundaries over a finite horizon. It also introduces a number of  new Volterra-type  integral equations based on an OU underlying process.  To the best of our knowledge, the integral equations and analytical results for optimal pair trading herein are new. Our integral equation approach has been recently applied to the valuation of swing options with  multiple stopping opportunities, see \cite{Kita1}.  Modeling the spread between the futures and spot by a Brownian bridge,   \cite{futuresDaiKwok} also consider a chooser option embedded in the trader's timing  to enter the market.    

The paper is organized as follows. We formulate the optimal trading
problem in Section 2.  The  solutions for the three trading problems and analyses of the optimal timing strategies  are presented in Sections 3, 4, and 5, respectively.    Section 6  is included to discuss the incorporation of transaction costs and associated challenges.
 \vs{6pt}

\section{Problem overview}
We fix a finite trading horizon $[0,T]$, and  filtered probability space $(\Omega, \mathcal{F} , (\mathcal{F}_t), \mathbb{P})$, where $\mathbb{P}$ is a subjective probability measure held by the trader, and $(\mathcal{F}_t)_{0\le t\le T}$ is the filtration to which    every process defined herein  is adapted. 

We  model the price  spread $X$ by the  Ornstein-Uhlenbeck (OU) process introduced by \cite{Ornstein1930}:
\begin{align}\label{OU}\hs{5pc}
dX_t=\mu(\theta\m X_t)dt+\sigma dB_t, \quad X_0=x,
\end{align}
where $B$ is a standard Brownian motion, and the parameters $\mu, \sigma>0$, and $\theta\in \R$, represent the speed of mean reversion, volatility, and long-run mean of the process, respectively. We refer the reader to, e.g.,    Section 2.1 of \cite{LeungLi2014OU}, for the detailed maximum likelihood   procedure   to estimate the parameters for this model with  backtested examples of spread trading. 

The solution to \eqref{OU} is
well-known as
\begin{align}\label{Sol}\hs{5pc}
X_t=x\e^{-\mu t}+\theta(1\m \e^{-\mu t})+\sigma\int_0^t \e^{-\mu(t-s)}dB_s, \qquad t\ge 0.
\end{align}
 At any fixed time $t$, the random variable $X_t$ has a normal distribution with the probability density function
\begin{align}\label{pdf}\hs{5pc}
p(\wt{x};t,x)=\frac{1}{\sqrt{2 \pi\; var(t,x)}}\,\e^{-\frac{(\wt{x}-m(t,x))^2}{2\;var(t,x)}}\,,
\end{align}
where the mean $m(t,x)$ and variance $var(t,x)$ functions are given by
\begin{align}\label{m}\hs{5pc}
m(t,x)&=x\e^{-\mu t}+\theta(1\m \e^{-\mu t}),\\
\label{var}var(t,x)&=\frac{\sigma^2}{2\mu}(1-\e^{-2\mu t}),
\end{align}
for $(t,x) \in \R_+\times \R$. The infinitesimal operator of $X$ is given as
\begin{align}\label{OU-operator}\hs{5pc}
\L_X F(x)=\mu(\theta\m x)F'(x)+\frac{\sigma^2}{2} F''(x)
\end{align}
for $x\in \R$ and $F\in C^2 (\R)$.
\vs{2pt}

The trader seeks to establish  a position and subsequently closes it by time $T>0$. We analyze three trading  strategies: (i) the 
long-short strategy, whereby the trader longs the spread first and later reverses the position to close (see Section \ref{ls}); (ii) short-long strategy, whereby the trader shorts the spread to start and then  close by taking the opposite (long) position (see Section \ref{sl}), and (iii)
 the chooser strategy, i.e. the trader  can   take either a  long or short position in the spread when entering the market  (see Section \ref{o}), and subsequently liquidates by taking the opposite position. For each trading problem, if it is optimal for the trader not  to trade at all, then in the absence of trading fees we represent it as if she/he enters and exits simultaneously at $T$,  giving a zero return. In other words, the trader must enter and exit by time  $T$. In the presence of transaction costs, simultaneous entry and exit will generate a strictly negative return, so the problem is more delicate and difficult (see Section 6).

\section{Optimal long-short strategy}\label{ls}
The trader who  enters the market    by taking a long position and subsequently closes it before time $T$ faces the following optimal double stopping problem:
 \begin{equation} \label{l-s-problem-1} \hs{7pc}
V^1(t,x)=\sup \limits_{0\le\tau\le\zeta\le T-t}\EE\left[\e^{-r\zeta} X^x_\zeta\m \e^{-r\tau} X^x_\tau\right],
 \end{equation}
defined at current time $t\in [0,T)$ and spread value $x\in \R$. Here, 
  $X^x$ represents   the process $X$ starting at $X^x_0=x$,   $r>0$ is the trader's discount  rate, and
 the supremum is taken over all pairs of $\cF^{X}$- stopping times $(\tau,\zeta)$ such that $\tau\le \zeta\le T-t$.
The time $\tau$ represents the strategy for buying the spread and $\zeta$ is the liquidation time.

Using standard arguments (see  e.g.   \cite{Carmona2008}), the problem \eqref{l-s-problem-1} can be reduced to  a sequence of two   optimal single stopping problems. Precisely, the optimal liquidation timing problem is represented by
 \begin{equation} \label{l-s-problem-2} \hs{7pc}
V^{1,L}(t,x)=\sup \limits_{0\le\zeta\le T-t}\EE \left[\e^{-r\zeta} X^x_\zeta\right].
 \end{equation} This value function $V^{1,L}$ represents the maximum expected value of a long position in the spread $X$, but the trader will need to pay the spread value for this position. Therefore, the difference between this value function and the spread value is viewed as the reward the trader received upon entry. Therefore, the trader's optimal entry timing problem is given by 
  \begin{equation} \label{l-s-problem-3} \hs{7pc}
V^{1,E}(t,x)=\sup \limits_{0\le\tau\le T-t}\EE \left[\e^{-r\tau}(V^{1,L}(t\p \tau,X^x_\tau)\m X^x_\tau)\right].
 \end{equation}
We have the equality: $V^1=V^{1,E}$, and  that the optimal stopping times in problems \eqref{l-s-problem-2}-\eqref{l-s-problem-3} are optimal for  the original problem  in  \eqref{l-s-problem-1}. We first solve the problem \eqref{l-s-problem-2} in Section~\ref{ls-exit}  and then the problem
 \eqref{l-s-problem-3} in Section~\ref{ls-entry}.
\vs{6pt}

 \subsection{Optimal exit problem}\label{ls-exit}
We now discuss how to represent the optimal stopping  problem \eqref{l-s-problem-2} as a free boundary problem, which is then analyzed   using the local time-space calculus (see \cite{Peskir2005a}).  First, using that the payoff function in \eqref{l-s-problem-2} is continuous and standard arguments (see e.g. Corollary 2.9 (finite horizon) with Remark 2.10 in \cite{Peskir2006}), we define the  continuation and exit regions:
\begin{align} \label{C1} \hs{5pc}
&\cC^{1,L}= \{\, (t,x)\in[0,T)\! \times\! \R:V^{1,L}(t,x)>x\, \}, \\[3pt]
 \label{D1}&\cD^{1,L}= \{\, (t,x)\in[0,T)\! \times\! \R:V^{1,L}(t,x)=x\, \},
 \end{align}
which are linked to  the optimal exit time in \eqref{l-s-problem-2} given  by
\begin{align} \label{OST} \hs{5pc}
\zeta_*^{1,L}=\inf\ \{\ 0\leq s\leq T\m t:(t\p s,X^x_{s})\in\cD^{1,L}\ \}.
 \end{align}

Let us define  the   function
\begin{align} \label{L}\hs{2pc}
\cK^{1,L}(u,x,z):=&-\e^{-ru}\EE \big[H^{1,L}(X^x_u) I(X^x_u \ge z)\big]\\=&-\e^{-ru}\int_z^{\infty} H^{1,L}(\wt{x})\,p(\wt{x};u,x)\,d\wt{x}
 \end{align}
 for $u\ge 0$ and $x,z\in\R$, where the normal density $p$ is given in \eqref{pdf} and the function $H^{1,L}$ is an affine function in $x$ defined by
 \begin{align} \label{H-1} \hs{6pc}
H^{1,L}(x)=-(\mu\p r)x+\mu\theta.
\end{align}

 The main result of this section is the integral equation representation of the free boundary associated with the optimal exit time $\zeta_*^{1,L}$.
\begin{theorem}\label{th:1}
The optimal stopping time  for \eqref{l-s-problem-2} is given   by
\begin{align} \label{OST-2} \hs{5pc}
\zeta_*^{1,L}=\inf\ \{\ 0\leq s\leq T\m t:X^x_{s}\ge b^{1,L}(t\p s) \ \}.
 \end{align}
 The function  $b^{1,L}(\cdot)$ is  the optimal exit boundary corresponding to  \eqref{l-s-problem-2}, and it can be characterized as the unique solution to a nonlinear integral equation of Volterra type, that is, 
\begin{align}\label{th-1-2} \hs{4pc}
b^{1,L}(t)=\e^{-r(T-t)}m(T\m t,b^{1,L}(t)) +\int_0^{T-t} \cK^{1,L}(u,b^{1,L}(t),b^{1,L}(t\p u))du
\end{align}
for $t\in[0,T]$ in the class of continuous decreasing functions $t\mapsto b^{1,L}(t)$ with $b^{1,L}(T)=x^*=\mu\theta/(\mu+r)$.
The value function $V^{1,L}$ in \eqref{l-s-problem-2} admits the  representation
\begin{align}\label{th-1-1} \hs{4pc}
V^{1,L}(t,x)=\e^{-r(T-t)}m(T\m t,x) +\int_0^{T-t}\cK^{1,L}(u,x,b^{1,L}(t\p u))du
\end{align}
for $t\in[0,T]$ and $x\in \R$.
\end{theorem}
\begin{proof}

The proof is provided in several steps.
\vs{6pt}

1. \emph{Continuity of $V^{1,L}$}. Here we show that the price function $V^{1,L}$ is continuous on $[0,T)\times \R$.
First we note that due to the linearity of the payoff function we obtain the convexity of $x\mapsto V^{1,L}(t,x)$ for fixed $t\in[0,T)$.  Therefore,  it follows that $x\mapsto V^{1,L}(t,x)$ is continuous for given $t\in [0,T)$. Thus to prove that
 $V^{1,L}$ is continuous on $[0,T)\times \R$ it is enough to show that $t\mapsto V^{1,L}(t,x)$ is continuous uniformly over $[0,T]$ for each $x\in \R$ given and fixed. For this, take any $t_1<t_2$ in $[0,T]$
and let $\zeta_1$ be an optimal stopping time for $V^{1,L}(t_1,x)$.  Setting $\zeta_2=\zeta_1\wedge (T-t_2)$ and using that
$t\mapsto V^{1,L}(t,x)$ is decreasing on $[0,T]$, we have
\begin{align} \label{cont-1} \hs{3pc}
0\le V^{1,L}(t_1,x)-V^{1,L}(t_2,x)\le \EE \left[\e^{-r\zeta_1} X^x_{\zeta_1}\right]-\EE \left[\e^{-r\zeta_2} X^x_{\zeta_2}\right]
\le \EE \left[|X^x_{\zeta_1}\m X^x_{\zeta_2}|\right].
\end{align}
Letting first $t_2-t_1\rightarrow0$ and using $\zeta_1-\zeta_2\rightarrow0$ we see that $V^{1,L}(t_1,x)-V^{1,L}(t_2,x)\rightarrow0$ by the dominated convergence theorem. This shows that
$t\mapsto V^{1,L}(t,x)$ is continuous uniformly over $[0,T]$, and the proof of the initial claim is complete.
\vs{6pt}

2. Now we obtain some initial insights into the structure of the exit region $\cD^{1,L}$.
 For this, we use Ito's formula to see that
\begin{align} \label{Ito} \hs{6pc}
 \EE \left[\e^{-r\zeta}X^{x}_\zeta\right]=\;x+\EE\left[ \int_0^\zeta \e^{-rs}H^{1,L}(X^{x}_s)ds\right]
 \end{align}
for $x\in\R$ and any stopping time $\zeta$ where
the function $H^{1,L}$ was defined in \eqref{H-1}.

The function $H^{1,L}$ is strictly decreasing with single root $x^*=\mu\theta/(\mu+ r)$.
 The equation \eqref{Ito} shows that it is not optimal to exit the position when $X_t< x^*$ as $H^{1,L}(X_t)>0$ in this region and thus the integral term
on the right-hand side of \eqref{Ito} is positive. For this one can make use of the first exit time from a
sufficiently small time-space ball centred at the point where $H^{1,L}>0$.
Another implication of \eqref{Ito} is that the exit region is non-empty for all $t\in[0,T)$, as for large $x\uparrow \infty$ the integrand $H^{1,L}$
is very negative and thus due to the lack of time to compensate the negative $H^{1,L}$, it is optimal to exit at once.
\vs{6pt}

3. \emph{Optimal exit boundary}. Next we prove further properties of the exit region $\cD^{1,L}$ and define the optimal exit boundary.

$(i)$ As the payoff function in \eqref{l-s-problem-2} is time-independent and the process $X$ is time-homogeneous,
it follows that the map $t\mapsto V^{1,L}(t,x)$ is non-increasing on $[0,T]$ for each $x\in \R$ so that $V^{1,L}(t_1,x)\m x\ge V^{1,L}(t_2,x)\m x\ge0$
for $0\le t_1 <t_2 <T $ and $x\in\R$. Now if we take a point $(t_1,x)\in \cD^{1,L}$, i.e. $V^{1,L}(t_1,x)=x$, then $(t_2,x)\in \cD^{1,L}$ as well, which shows that the exit region expands when $t$ increases.
In other words, $\cD^{1,L}$ is right-connected.
\vs{2pt}

$(ii)$ Now let us take $t>0$, $x>y$, and we denote by $\zeta=\zeta(t,x)$ the optimal stopping time
for $V^{1,L}(t,x)$. Then using \eqref{Ito} we have  
\begin{align} \label{Ito-a} \hs{5pc}
V^{1,L}(t,x)-V^{1,L}(t,y) &\le  \EE \left[\e^{-r\zeta} X^{x}_\zeta\right]- \EE \left[\e^{-r\zeta}X^{y}_\zeta\right]\\
&=x-y+\EE \left[\int_0^\zeta \e^{-rs}\left(H^{1,L}(X^{x}_s)\m H^{1,L}(X^y_s)\right)ds\right]\nonumber\\
&\le x-y\nonumber
 \end{align}
where we used that $H^{1,L}$ is decreasing and $X^x_{\cdot}\ge X^y_{\cdot}$ by \eqref{Sol}. Now if we let $(t,y)\in \cD^{1,L}$, i.e. $V^{1,L}(t,y)=y$, we have   $V^{1,L}(t,x)=x$, i.e. $(t,x)\in \cD^{1,L}$.
Therefore we obtain an up-connectedness of the exit region $\cD^{1,L}$.
\vs{2pt}

$(iii)$ From $(i)$-$(ii)$ and paragraph 2 above we can conclude that
there exists an optimal exit boundary $b^{1,L}:[0,T]\rightarrow \R$ such that
\begin{align}   \hs{5pc}
\zeta_*^{1,L}=\inf\ \{\ 0\leq s\leq T\m t:X^x_{s}\ge b^{1,L}(t\p s) \ \}
 \end{align}
is optimal in \eqref{l-s-problem-2} and $x^*<b^{1,L}(t)<\infty$ for $t\in[0,T)$. Moreover, $b^{1,L}$ is decreasing on $[0,T)$.
\vs{6pt}

4. \emph{Smooth-fit}. Now we prove that the smooth-fit condition along the boundary $b^{1,L}$ holds
 \begin{align}\label{SF}\hs{5pc}
V^{1,L}_x (t, b^{1,L}(t)-)=V^{1,L}_x(t,b^{1,L}(t)+)=1
\end{align}
for all $t\in[0,T)$.
\vs{2pt}

$(i)$ First let us fix a point $(t,x)\in [0,T)\times\R$ lying on the boundary $b^{1,L}$ so that $x=b^{1,L}(t)$.
Then we have
\begin{align} \label{SF-1} \hs{5pc}
\frac{V^{1,L}(t,x)-V^{1,L}(t,x\m\eps)}{\eps}&\le \frac{x-(x\m\eps)}{\eps}=1
\end{align}
and taking the limit as $\eps\downarrow 0$, we get
\begin{align} \label{SF-2} \hs{5pc}
V^{1,L}_x (t,x-)\le 1
\end{align}
 where the left-hand derivative exists by the convexity of $x\mapsto V^{1,L}(t,x)$ on $\R$ for any fixed $t\in[0,T)$.
\vs{2pt}

$(ii)$ To prove the reverse inequality, we set $\zeta_\eps=\zeta_\eps(t,x\m\eps)$ as an optimal stopping time for $V^{1,L}(t,x-\eps)$.
Using that $X$ is a regular diffusion and $t\mapsto b^{1,L}(t)$ is decreasing,  we see that 
$\zeta_\eps\to0$ as $\eps\to0$ $\PP$-a.s. We get  
 \begin{align}\label{SF-3}\hs{4pc}
\frac{1}{\eps}\Big(V^{1,L}(t,x)-V^{1,L}(t,x\m\eps)\Big)
 \ge\frac{1}{\eps}\EE \left[\e^{-r\zeta_{\eps}}\left( X_{\zeta_{\eps}}^{x}\m X_{\zeta_{\eps}}^{x-\eps}\right) \right]=\EE \left[\e^{-(r+\mu)\zeta_{\eps}}\right]
\end{align}
where we used the solution \eqref{Sol} for $X$. Clearly, the right-hand side of \eqref{SF-3} goes to 1 as $\eps\to0$.
Thus taking the limits as $\eps\to0$, we get the inequality
\begin{align} \label{SF-4}\hs{7pc}
V^{1,L}_x (t,x-)\ge 1
\end{align}
for $t\in[0,T)$. Combining \eqref{SF-2} and \eqref{SF-4}, we  obtain \eqref{SF}.
\vs{6pt}

5. \emph{Continuity of $b^{1,L}$}. Here we prove that the boundary $b^{1,L}$ is continuous on $[0,T]$ and that $b^{1,L}(T-)=x^*$.
The proof is provided in 3 steps and follows the approach proposed by \cite{DeA1}.

$(i)$ We first show that $b^{1,L}$ is right-continuous. Let us fix $t\in [0,T)$ and take a
sequence $t_n \downarrow t$ as $n\rightarrow \infty$. As $b^{1,L}$ is decreasing, the right-limit $b^{1,L}(t+)$ exists and
$(t_n, b^{1,L}(t_n))$ belongs to $\cD^{1,L}$ for all $n\ge 1$. Recall that $\cD^{1,L}$ is closed so that $(t_n,b^{1,L}(t_n))\to (t,b^{1,L}(t+))\in \cD^{1,L}$ as $n\to \infty$ and we may conclude that $b^{1,L}(t+)\ge b^{1,L}(t)$. The fact that $b^{1,L}$
is decreasing gives the reverse inequality and thus $b^{1,L}$ is right-continuous as claimed.
\vs{+2pt}

$(ii)$ Now we prove that $b^{1,L}$ is also left-continuous. Assume that there exists $t_0\in(0,T)$ such that $b^{1,L}(t_0-)>b^{1,L}(t_0)$.
Let us set $x_1 = b^{1,L}(t_0)$ and $x_2 = b^{1,L}(t_0-)$ so that $x_1 < x_2$. For $\eps \in (0,(x_2- x_1)/2)$ given and fixed, let $\varphi_\eps : (-\infty,\infty) \rightarrow [0,1]$ be a $C^\infty$- function satisfying (i) $\varphi_\eps (x) = 1$ for $x\in [x_1 \p\eps, x_2 \m\eps]$ and
(ii) $\varphi_\eps (x) = 0$ for $x\in (-\infty,x_1\p \eps/2]\cup [x_2\m\eps/2,\infty)$. Letting $\L^*_X $ denote the adjoint of
$\L_X$, recalling that $t\rightarrow  V^{1,L}(t,x)$ is decreasing on $[0,T]$ and that $V^{1,L}_t\p \L_X V^{1,L}  \m rV^{1,L}= 0$ on $\cC^{1,L}$, we find integrating by parts (twice) that
\begin{align}\label{cont01}\hs{1pc}
0\ge \int^{x_2}_{x_1}{\varphi(x)V^{1,L}_t(t_0 \m \delta,x)dx}=-\int^{x_2}_{x_1}{V^{1,L}(t_0 \m \delta,x)\left(\L_X^*\varphi(x)\m r\varphi(x)\right)dx}
\end{align}
for $\delta\in (0,t_0\wedge  (\eps/2))$ so that $\varphi_\eps (x_2\m\delta) = \varphi'_\eps (x_2\m\delta)  = 0$ as needed. Letting $\delta \downarrow 0$ it follows using the dominated convergence theorem and integrating by parts (twice) that
\begin{align}\label{cont02}\hs{1pc}
0&\ge -\int^{x_2}_{x_1}{V^{1,L}(t_0,x)\left(\L_X^*\varphi(x)\m r\varphi(x)\right)dx}
=-\int^{x_2}_{x_1}{x\left(\L_X^*\varphi(x)\m r\varphi(x)\right)dx}\\
&=-\int^{x_2}_{x_1}{\left(\L_X x\m rx\right)\varphi(x) dx}=-\int^{x_2}_{x_1}{H^{1,L}(x)\varphi(x)dx}.\nonumber
\end{align}
Letting $\eps\downarrow  0$ we obtain
\begin{align}\label{cont03}\hs{5pc}
0\ge -\int^{x_2}_{x_1}{H^{1,L}(x)dx}>0
\end{align}
as $x\rightarrow H^{1,L}(x)$ is strictly negative on $(x_1,x_2]$. We thus have a contradiction and therefore we may conclude that $b^{1,L}$ is continuous on $[0, T )$ as claimed.
\vs{+2pt}

$(iii)$ To prove that $b^{1,L}(T-)=x^*$ we can use the same arguments as those in $(ii)$ above with $t_0=T$ and suppose that $b^{1,L}(T-)>x^*$.
\vs{6pt}

6. \emph{Free-boundary problem}. The facts proved in paragraphs 1-5 above and standard arguments based on the strong Markov property (see e.g. \cite{Peskir2006}) lead to the following free-boundary problem for the value function $V^{1,L}$ and unknown boundary $b^{1,L}$:
\begin{align} \label{PDE} \hs{5pc}
&V^{1,L}_t \p\L_X V^{1,L} \m rV^{1,L}=0 &\hs{-30pt}\text{in}\;  \cC^{1,L}\\
\label{IS}&V^{1,L}(t,b^{1,L}(t))=b^{1,L}(t)&\hs{-30pt}\text{for}\; t\in[0,T)\\
\label{SP}&V^{1,L}_x (t,b(t))=1 &\hs{-30pt}\text{for}\; t\in[0,T) \\
\label{FBP1}&V^{1,L}(t,x)>x &\hs{-30pt}\text{in}\; \cC^{1,L}\\
\label{FBP2}&V^{1,L}(t,x)=x &\hs{-30pt}\text{in}\; \cD^{1,L}
\end{align}
where the continuation set $\cC^{1,L}$ and the exit region $\cD^{1,L}$ are given by
\begin{align} \label{CC1} \hs{5pc}
&\cC^{1,L}= \{\, (t,x)\in[0,T)\! \times\! \R:x<b^{1,L}(t)\, \} \\[3pt]
 \label{DD1}&\cD^{1,L}= \{\, (t,x)\in[0,T)\! \times\! \R:x\ge b^{1,L}(t)\, \}.
 \end{align}
The following properties of $V^{1,L}$ and $b^{1,L}$ were also verified above:
\begin{align} \label{Prop-1} \hs{5pc}
&V^{1,L}\;\text{is continuous on}\; [0,T]\times\R\\
\label{Prop-2}&V^{1,L}\;\text{is}\; C^{1,2}\;\text{on}\; \cC^{1,L}\\
\label{Prop-3}&x\mapsto V^{1,L}(t,x)\;\text{is increasing and convex on $\R$ for each $t\in[0,T]$}\\
\label{Prop-4}&t\mapsto V^{1,L}(t,x)\;\text{is decreasing on $[0,T]$ for each $x\in \R$}\\
\label{Prop-5}&t\mapsto b^{1,L}(t)\;\text{is decreasing and continuous on $[0,T]$ with}\; b^{1,L}(T-)=x^*.
\end{align}

7. \emph{Integral equation}.
 We clearly have that the following conditions hold:
$(i)$ $V^{1,L}$ is $C^{1,2}$ on $\cC^{1,L}\cup \cD^{1,L}$;
$(ii)$ $b^{1,L}$ is of bounded variation (due to monotonicity);
 $(iii)$ $V^{1,L}_t\p\L_{X}V^{1,L}-rV^{1,L}$ is locally bounded;
 $(iv)$ $x\mapsto V^{1,L}(t,x)$ is convex;
$(v)$ $t\mapsto V^{1,L}_x (t,b^{1,L}(t)\pm)$ is continuous (recall \eqref{SP}). Hence we can apply the local time-space formula on curves \cite{Peskir2005a}) for $\e^{-rs}V^{1,L}(t\p s,X^x_s)$, along with   \eqref{PDE}, \eqref{SP}, and \eqref{FBP2},  to get 
\begin{align} \label{th-3} \hs{1pc}
&\e^{-rs}V^{1,L}(t\p s,X^x_s)\\
=\;&V^{1,L}(t,x)+M_s\nonumber\\
 &+ \int_0^{s}
\e^{-ru}\left(V^{1,L}_t \p\L_X V^{1,L}
\m rV^{1,L}\right)(t\p u,X^x_u)
 I(X^x_u \neq b^{1,L}(t\p u))du\nonumber\\
 &+\frac{1}{2}\int_0^{s}
\e^{-ru}\left(V^{1,L}_x (t\p u,X^x_u +)-V^{1,L}_x (t\p u,X^x_u -)\right)I\big(X^x_u=b^{1,L}(t\p u)\big)d\ell^{b^{1,L}}_u(X^x)\nonumber\\
  =\;&V^{1,L}(t,x)+M_s +\int_0^{s}
\e^{-ru}H^{1,L}(X^x_u) I(X^x_u \ge b^{1,L}(t\p u))du\nonumber
  \end{align}
where    $M=(M_s)_{s\ge 0}$ is the martingale part, and   $(\ell^{b^{1,L}}_t(X^x))_{t\ge 0}$ is the local time process of $X^x$ at the boundary $b^{1,L}$, given by 
\begin{align} \label{Tanaka} \hs{3pc}
 \ell^{b^{1,L}}_t (X^x):=\PP-\lim_{\eps \downarrow 0}\frac{1}{2\eps}\int_0^{t} I(b^{1,L}(t\p u)\m\eps<X^x_u<b^{1,L}(t\p u)\p\eps)d\left \langle X,X \right \rangle_u.
 \end{align}
 Now
upon letting $s=T\m t$, taking the expectation $\EE$, using the optional sampling theorem for $M$, rearranging terms, noting that
$V^{1,L}(T,x)=x$ for all $x\in \R$ and recalling \eqref{m}, we get \eqref{th-1-1}.
The integral equation \eqref{th-1-2} is obtained by inserting $x=b^{1,L}(t)$ into \eqref{th-1-1} and using \eqref{IS}.
\vs{6pt}

8. \emph{Uniqueness of the solution}. The proof of that $b^{1,L}$ is the unique solution to the equation \eqref{th-1-2} in the class of continuous decreasing functions $t\mapsto b^{1,L}(t)$ is based on arguments originally employed by \cite{Peskir2005b} and omitted here.

\end{proof}

\emph{Numerical algorithm for solution to integral equation}. We proved above that $b^{1,L}$ is the unique solution to the integral  equation \eqref{th-1-2}. Even though  this equation
cannot be solved analytically, it can    be solved  numerically in a straightforward and efficient manner, as we illustrate below and refer to   Chapter 8 of  \cite{DetempleBook} for more details.  In order to numerically solve the integral equation, it is crucial to be able to  compute $\cK^{1,L}$ efficiently. Fortunately, we have the closed-form expression for the function $\cK^{1,L}$ since  the (marginal) distribution of $X_t$ is Gaussian.

Let $N$ be the number of time discretizations, and set  $h=T/N$  and  $t_k=kh$ for $k=0,1,...,N$. This leads to 
the following discrete approximation of the integral equation  \eqref{th-1-2}:
\begin{align}\label{Rem-1} \hs{2pc}
b^{1,L}(t_k)=\e^{-r(T-t)}m(T-t,b^{1,L}(t_k))+h\sum_{l=k}^{N-1} \cK^{1,L}\left(t_k,t_{l+1}\m t_k,b^{1,L}(t_k),b^{1,L}(t_{l+1})\right)
\end{align}
for $k= 0,1,...,N\m1$. Setting $k=N\m1$ and $b^{1,L}(t_N)=x^*$ we   solve the equation
\eqref{Rem-1} numerically and obtain the value of  $b^{1,L}(t_{N- 1})$. Setting $k=N\m2$ and using the values $b^{1,L}(t_{N- 1})$ and 
$b^{1,L}(t_{N})$, we   solve \eqref{Rem-1} numerically for the value $b^{1,L}(t_{N- 2})$. Continuing this  recursion we obtain all  $\{b^{1,L}(t_{N}),b^{1,L}(t_{N-1}),...,b^{1,L}(t_1),b^{1,L}(t_0)\}$  as approximations to  the continuous optimal boundary $b^{1,L}$ at the
points $T, T- h,..., h, 0$.

Finally, the value function \eqref{th-1-1} can be approximated as follows:
\begin{align}\label{Rem-2} \hs{2pc}
V^{1,L}(t_k,x)=\e^{-r(T-t)}m(T-t,x)+h\sum_{l=k}^{N-1} \cK^{1,L}\left(t_k,t_{l+1}\m t_k,x,b^{1,L}(t_{l+1})\right)
\end{align}
for $k= 0,1,...,N\m 1$ and $x\in \R$.

 \subsection{Optimal entry problem}\label{ls-entry}
Having solved for the optimal   timing to exit, we now turn to the optimal entry problem
 \begin{equation} \label{l-s-problem-4} \hs{5pc}
V^{1,E}(t,x)=\sup \limits_{0\le\tau\le T-t}\EE \left[\e^{-r\tau}\left(V^{1,L}(t\p\tau,X^x_\tau)-X^x_\tau \right)\right]
 \end{equation}
where the supremum is taken over all stopping times $\tau\in[0,T-t]$ of $X$. We define the payoff function $G^{1,E}(t,x)=V^{1,L}(t,x)-x$ for $(t,x)\in [0,T)\times \R$.
\vs{6pt}

 We tackle the problem \eqref{l-s-problem-4} using similar arguments as for \eqref{l-s-problem-2}.
We define the continuation and entry regions
\begin{align} \label{C-2} \hs{5pc}
&\cC^{1,E}= \{\, (t,x)\in[0,T)\! \times\! \R:V^{1,E}(t,x)>G^{1,E}(t,x)\, \} \\[3pt]
 \label{D-2}&\cD^{1,E}= \{\, (t,x)\in[0,T)\! \times\! \R:V^{1,E}(t,x)=G^{1,E}(t,x)\, \}.
 \end{align}
In turn,  the optimal exit time in \eqref{l-s-problem-4} is given by
\begin{align} \label{OST-2} \hs{5pc}
\tau^{1,E}_*=\inf\ \{\ 0\leq s\leq T\m t:(t\p s,X^x_{s})\in\cD^{1,E}\ \}.
 \end{align}
Let us define the function $\cK^{1,E}$ as
\begin{align} \hs{3pc}
\label{L-2}
\cK^{1,E}(u,x,z)&=-\e^{-ru}\EE \big[H^{1,E}(u, X^x_u) I(X^x_u \le z)\big]\notag\\
&=-\e^{-ru}\int_{-\infty}^z H^{1,E}(u, \wt{x})\,p(\wt{x};u,x)\,d\wt{x}
 \end{align}
 for $u\ge 0$ and $x,z\in\R$, where  \begin{align}\hs{3pc}H^{1,E}(t,x)=((\mu+ r)x-\mu\theta)I(x <b^{1,L}(t))
 \end{align} for $(t,x)\in[0,T)\times\R$.

 We now state the main theorem of this section. We do not provide full proof since it is very similar to the one in Theorem 3.1 above and outline only important details.
\begin{theorem}\label{th:2}
The optimal entry boundary $b^{1,E}$ in \eqref{l-s-problem-4} can be characterized as the unique solution to a nonlinear integral equation
\begin{align}\label{th-2-2} \hs{5pc}
V^{1,L}(t,b^{1,E}(t))\m b^{1,E}(t)=\int_0^{T-t} \cK^{1,E}(u,b^{1,E}(t),b^{1,E}(t\p u))du
\end{align}
for $t\in[0,T]$ in the class of continuous increasing functions $t\mapsto b^{1,E}(t)$ with $b^{1,E}(T)=x^*$.
 The value function $V^{1,E}$ in \eqref{l-s-problem-4} can be represented as
\begin{align}\label{th-2-1} \hs{5pc}
V^{1,E}(t,x)=\int_0^{T-t}\cK^{1,E}(u,x,b^{1,E}(t\p u))du
\end{align}
for $t\in[0,T]$ and $x\in \R$.
\end{theorem}
\begin{proof}

1. We use the local time-space formula on curves (\cite{Peskir2005a}) and the smooth-fit property \eqref{SP} to obtain
\begin{align} \label{Ito-2} \hs{5pc}
 \EE \left[\e^{-r\tau}G^{1,E}(t\p\tau,X^{x}_\tau)\right]=\;G^{1,E}(t,x)+\EE \left[\int_0^\tau \e^{-rs}H^{1,E}(t\p s,X^{x}_s)ds\right]
 \end{align}
for $t\in[0,T)$, $x\in\R$, any stopping time $\tau$ of process $X$  and where
the function $H^{1,E}$ is defined as  $H^{1,E}(t,x):=(\L_X G^{1,E} \m rG^{1,E})(t,x)$ for $(t,x)\in[0,T)\times\R$ and equals
\begin{align} \label{H-2} \hs{5pc}
H^{1,E}(t,x)=((\mu\p r)x-\mu\theta)I(x <b^{1,L}(t))
\end{align}
for $(t,x)\in[0,T)\times\R$.

The function $H^{1,E}$ is strictly increasing and linear in $x$ on $(-\infty,b^{1,L}(t))$ for fixed $t$ and has unique root $x^*=\mu\theta/(\mu\p r)<b^{1,L}(t)$. Hence, it is not optimal to enter into the position when $X_t>x^*$ and as the integral term
on the right-hand side of \eqref{Ito-2} is non-negative.
The equation \eqref{Ito-2} also gives that the entry region is non-empty for all $t\in[0,T)$, as for large negative $x\downarrow -\infty$ the integrand $H^{1,E}$
is very negative and thus it is optimal to enter immediately.
\vs{6pt}

2. We can prove similarly as in the previous section that the entry region $\cD^{1,E}$ is right-connected and down-connected.
Hence there exists an optimal entry boundary $b^{1,E}:[0,T]\rightarrow \R$ (see Figure 1) such that
\begin{align}  \hs{5pc}
&\tau^{1,E}_*=\inf\ \{\ 0\leq s\leq T\m t:X^x_{s}\le b^{1,E}(t\p s) \ \}
 \end{align}
is optimal in \eqref{l-s-problem-4} and $-\infty<b^{1,E}(t)<x^*$ for $t\in[0,T)$. Moreover, $b^{1,E}$ is increasing on $[0,T)$ and is bounded from below.
\vs{6pt}

3.  Standard methods based on the strong Markov property and arguments from the previous section lead to the following free-boundary problem for the value function $V^{1,E}$ and the boundary $b^{1,E}$:
\begin{align} \label{PDE-2} \hs{5pc}
&V^{1,E}_t \p\L_X V^{1,E} \m rV^{1,E}=0 &\hs{-30pt}\text{in}\;  \cC^{1,E}\\
\label{IS-2}&V^{1,E}(t,b^{1,E}(t))=V^{1,L}(t,b^{1,E}(t))\m b^{1,E}(t) &\hs{-30pt}\text{for}\; t\in[0,T)\\
\label{SP-2}&V^{1,E}_x(t,b^{1,E}(t))=V^{1,L}_x(t,b^{1,E}(t))\m 1&\hs{-30pt}\text{for}\; t\in[0,T) \\
\label{FBP1-2}&V^{1,E}(t,x)>G^{1,E}(t,x) &\hs{-30pt}\text{in}\; \cC^{1,E}\\
\label{FBP2-2}&V^{1,E}(t,x)=G^{1,E}(t,x) &\hs{-30pt}\text{in}\; \cD^{1,E}
\end{align}
where the continuation set $\cC^{1,E}$ and the entry region $\cD^{1,E}$ are given by
\begin{align} \label{CC-2} \hs{5pc}
&\cC^{1,E}= \{\, (t,x)\in[0,T)\! \times\! \R:x>b^{1,E}(t)\, \} \\[3pt]
 \label{DD-2}&\cD^{1,E}= \{\, (t,x)\in[0,T)\! \times\! \R:x\le b^{1,E}(t)\, \}.
 \end{align}
The following properties of $V^{1,E}$ and $b^{1,E}$ hold:
\begin{align} \label{Prop-1-2} \hs{5pc}
&V^{1,E}\;\text{is continuous on}\; [0,T]\times\R\\
\label{Prop-2-2}&V^{1,E}\;\text{is}\; C^{1,2}\;\text{on}\; \cC^{1,E}\\
\label{Prop-3-2}&x\mapsto V^{1,E}(t,x)\;\text{is convex on $\R$ for each $t\in[0,T]$}\\
\label{Prop-4-2}&t\mapsto V^{1,E}(t,x)\;\text{is decreasing on $[0,T]$ for each $x\in \R$}\\
\label{Prop-5-2}&t\mapsto b^{1,E}(t)\;\text{is increasing and continuous on $[0,T]$ with}\; b^{1,E}(T-)=x^*.
\end{align}

4. We then verify the conditions of local time-space formula and apply it for $\e^{-rs}V^{1,E}(t\p s,X^x_s)$ to obtain representation \eqref{th-2-1}. The integral equation  \eqref{th-2-1} is
derived by inserting $x=b^{1,E}(t)$ into \eqref{th-2-1}.

\end{proof}
\clearpage
\begin{figure}[t]
\begin{center}
\includegraphics[scale=0.7]{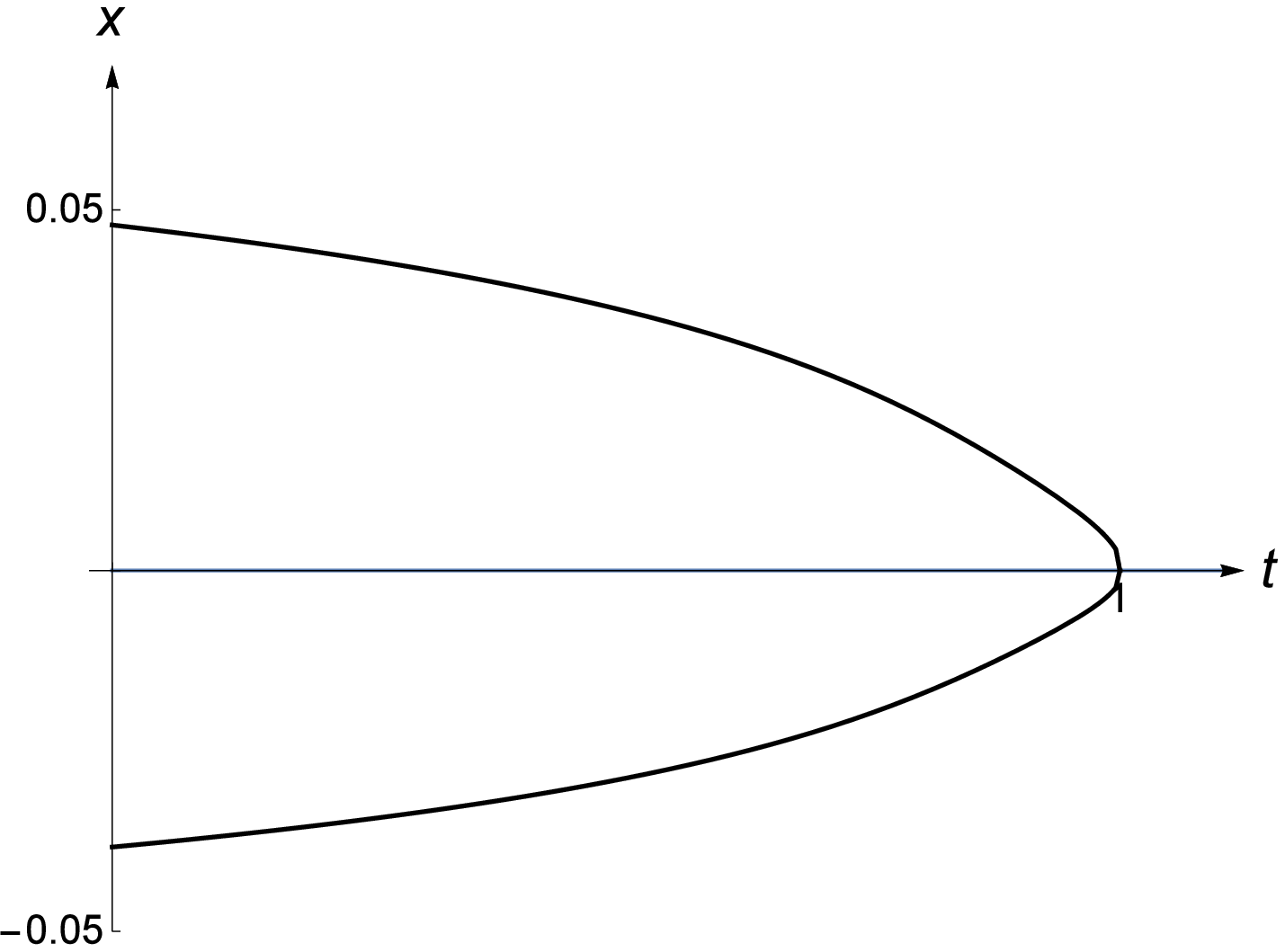}
\end{center}

{\par \leftskip=1.6cm \rightskip=1.6cm \small \ni \vs{-10pt}

\textbf{Figure 1.} The optimal entry boundary
 (lower) and the optimal exit boundary (upper) for the long-short strategy problem \eqref{l-s-problem-1} computed as the solutions to integral equations \eqref{th-2-2} and \eqref{th-1-2}, respectively.
The parameters are:  $T=1$ year, $r=0.01$, $\theta=0$, $\mu=16$, $\zeta=0.16.$  A time discretization with  $500$ steps for the interval $[0,T]$ is used.

\par} \vs{10pt}

\end{figure}

 \section{Optimal short-long strategy}\label{sl}
In this section, we consider the short-long strategy: short  the spread 
to open, long the spread to close the position. This problem is analogous to the one in Section 3, and thus we only  state   the optimal stopping problems and main results, and omit the proofs.

At      time $t\in [0,T)$ with  the current spread value $x\in \R$, the trader solves the  optimal double stopping problem 
 \begin{equation} \label{s-l-problem-1} \hs{7pc}
V^2(t,x)=\sup \limits_{0\le\zeta\le\tau\le T-t}\EE \left[\e^{-r\zeta} X^x_\zeta\m \e^{-r\tau} X^x_\tau\right]
 \end{equation}
where $X^x$ represents that the process $X$ starts from $X^x_0=x$, $r>0$ is the interest rate, and
 the supremum is taken over all pairs of $\cF^{X}$- stopping times $(\tau,\zeta)$ such that $\zeta\le\tau\le T-t$.
As in the previous section, $\tau$ is time for long position and $\zeta$ is the strategy for short position.

\begin{remark}  If $\theta=0$,  then $X$ and $-X$ have the same law. As such,   the problems  \eqref{l-s-problem-1} and  \eqref{s-l-problem-1} are symmetric, i.e.,
$V^{2}=-V^{1}$.
\end{remark}

We reduce \eqref{s-l-problem-1} into the  two single optimal stopping problems:
 \begin{equation} \label{s-l-problem-2} \hs{7pc}
V^{2,L}(t,x)=\inf \limits_{0\le\tau\le T-t}\EE \left[\e^{-r\tau} X^x_\tau\right]
 \end{equation}
which is the optimal exit problem under this strategy. After   solving \eqref{s-l-problem-2} we will turn to the entry problem, i.e. optimal timing to enter into the long position
  \begin{equation} \label{s-l-problem-3} \hs{4pc}
V^{2}(t,x)=V^{2,E}(t,x)=\sup \limits_{0\le\zeta\le T-t}\EE \left[\e^{-r\zeta}(X^x_\zeta \m V^{2,L}(t\p \zeta,X^x_\zeta))\right]
 \end{equation}
as at time $t\p\zeta$ we receive $X^x_\zeta$ and get the short position with the value $-V^{2,L}(t\p \zeta,X^x_\zeta)$.
\vs{2pt}

 As in the previous section, we   reduce \eqref{s-l-problem-2} to the free-boundary problem for the value function $V^{2,L}$ and optimal exit boundary $b^{2,L}$:
\begin{align}
\label{PDE-3} \hs{5pc}&V^{2,L}_t \p\L_X V^{2,L} \m rV^{2,L}=0 &\hs{-30pt}\text{in}\; \cC^{2,L}\\
\label{IS-3}&V^{2,L}(t,b^{2,L}(t))=b^{2,L}(t) &\hs{-30pt}\text{for}\; t\in[0,T)\\
\label{SP-3}&V^{2,L}_x (t,b^{2,L}(t))=1 &\hs{-30pt}\text{for}\; t\in[0,T) \\
\label{FBP1-3}&V^{2,L}(t,x)<x &\hs{-30pt}\text{in}\; \cC^{2,L}\\
\label{FBP2-3}&V^{2,L}(t,x)=x  &\hs{-30pt}\text{in}\; \cD^{2,L}
\end{align}
where the continuation set $\cC^{2,L}$ and the exit region $\cD^{2,L}$ are given by
\begin{align} \label{CC-3} \hs{5pc}
&\cC^{2,L}= \{\, (t,x)\in[0,T)\! \times\! \R:x>b^{2,L}(t)\, \} \\[3pt]
 \label{DD-3}&\cD^{2,L}= \{\, (t,x)\in[0,T)\! \times\! \R:x\le b^{2,L}(t)\, \}.
 \end{align}
The value function $V^{2,L}$ and optimal boundary  $b^{2,L}$ admit the following properties:
\begin{align} \label{Prop-1-3} \hs{5pc}
&V^{2,L}\;\text{is continuous on}\; [0,T]\times\R\\
\label{Prop-2-3}&V^{2,L}\;\text{is}\; C^{1,2}\;\text{on}\; \cC^{2,L}\\
\label{Prop-3-3}&x\mapsto V^{2,L}(t,x)\;\text{is increasing and concave on $\R$ for each $t\in[0,T]$}\\
\label{Prop-4-3}&t\mapsto V^{2,L}(t,x)\;\text{is increasing on $[0,T]$ for each $x\in \R$}\\
\label{Prop-5-3}&t\mapsto b^{2,L}(t)\;\text{is increasing and continuous on $[0,T]$ with}\; b^{2,L}(T-)=x^*.
\end{align}

 The main theorem is stated as follows.
\begin{theorem}\label{th:3}
The value function $V^{2,L}$ in \eqref{s-l-problem-2} has the following representation
\begin{align}\label{th-3-1} \hs{5pc}
V^{2,L}(t,x)=\e^{-r(T-t)} m(T-t,x) +\int_0^{T-t}\cK^{2,L}(u,x,b^{2,L}(t\p u))du
\end{align}
for $t\in[0,T]$ and $x\in \R$. The optimal exit boundary $b^{2,L}$ in \eqref{s-l-problem-2} can be characterized as the unique solution to a nonlinear integral equation
\begin{align}\label{th-3-2} \hs{5pc}
b^{2,L}(t)=\e^{-r(T-t)} m(T-t,b^{2,L}(t)) +\int_0^{T-t} \cK^{2,L}(u,b^{2,L}(t),b^{2,L}(t\p u))du
\end{align}
for $t\in[0,T]$ in the class of continuous increasing functions $t\mapsto b^{2,L}(t)$ with $b^{2,L}(T)=x^*$.
and where the function $\cK^{2,L}$ is defined as
\begin{align}\label{K} \hs{5pc}
\cK^{2,L}(u,x,z)=-\e^{-ru}\EE \big[H^{1,L}(X^x_u) I(X^x_u \le z)\big]
 \end{align}
 for $u\ge 0$, $x,z\in\R$ and $H^{1,L}$ is given in \eqref{H-1}.
 \end{theorem}

Applying   the results for exit problem, we now analyze   the  following free-boundary problem for the optimal  entry value function $V^{2,E}$ in  \eqref{s-l-problem-3}, and the associated boundary $b^{2,E}$:
\begin{align} \label{PDE-4} \hs{6pc}
&V^{2,E}_t \p\L_X V^{2,E} \m rV^{2,E}=0 &\hs{-30pt}\text{in}\;  \cC^{2,E}\\
\label{IS-4}&V^{2,E}(t,b^{2,E}(t))=b^{2,E}(t)\m V^{2,L}(t,b^{2,E}(t)) &\hs{-30pt}\text{for}\; t\in[0,T)\\
\label{SP-4}&V^{2,E}_x (t,b^{2,E}(t))=1\m V^{2,L}_x(t,b^{2,E}(t))&\hs{-30pt}\text{for}\; t\in[0,T) \\
\label{FBP1-4}&V^{2,E}(t,x)>x\m V^{2,L}(t,x) &\hs{-30pt}\text{in}\; \cC^{2,E}\\
\label{FBP2-4}&V^{2,E}(t,x)=x\m V^{2,L}(t,x) &\hs{-30pt}\text{in}\; \cD^{2,E}
\end{align}
where the continuation  and entry regions, respectively, are given by
\begin{align} \label{C-4} \hs{5pc}
&\cC^{2,E}= \{\, (t,x)\in[0,T)\! \times\! \R:x<b^{2,E}(t)\, \} \\[3pt]
 \label{D-4}&\cD^{2,E}= \{\, (t,x)\in[0,T)\! \times\! \R:x\ge b^{2,E}(t)\, \}.
 \end{align}
The following properties of $V^{2,E}$ and $b^{2,E}$ hold:
\begin{align} \label{Prop-1-4} \hs{6pc}
&V^{2,E}\;\text{is continuous on}\; [0,T]\times\R\\
\label{Prop-2-4}&V^{2,E}\;\text{is}\; C^{1,2}\;\text{on}\; \cC^{2,E}\\
\label{Prop-3-4}&x\mapsto V^{2,E}(t,x)\;\text{is convex on $\R$ for each $t\in[0,T]$}\\
\label{Prop-4-4}&t\mapsto V^{2,E}(t,x)\;\text{is decreasing on $[0,T]$ for each $x\in \R$}\\
\label{Prop-5-4}&t\mapsto b^{2,E}(t)\;\text{is decreasing and continuous on $[0,T]$ with}\; b^{2,E}(T-)=x^*.
\end{align}

To prepare the following result, we define the function
\begin{align} \hs{6pc}
\label{K-2}
\cK^{2,E}(u,x,z)=-\e^{-ru}\EE \big[H^{1,L}(X^x_u) I(X^x_u \ge z)\big]
 \end{align}
 for $t, u\ge 0$, $x,z\in\R$, with   $H^{1,L}$   defined in \eqref{H-1}.

\begin{theorem}\label{th:4}
The value function $V^{2,L}$ in \eqref{s-l-problem-3}  admits the integral representation:
\begin{align}\label{th-4-1} \hs{6pc}
V^{2,E}(t,x)=\int_0^{T-t}\cK^{2,E}(u,x,b^{2,E}(t\p u))du
\end{align}
for $t\in[0,T]$ and $x\in \R$. The optimal entry boundary $b^{2,E}$ in \eqref{s-l-problem-3} can be characterized as the unique solution to a nonlinear integral equation
\begin{align}\label{th-4-2} \hs{6pc}
b^{2,E}(t)\m V^{2,L}(t,b^{2,E}(t))=\int_0^{T-t} \cK^{2,E}(u,b^{2,E}(t),b^{2,E}(t\p u))du
\end{align}
for $t\in[0,T]$ in the class of continuous decreasing functions $t\mapsto b^{2,E}(t)$ with $b^{2,E}(T)=x^*$.
\end{theorem}

\section{Chooser strategy}\label{o}
 In this section, the trader in  the spread trading problem  can choose whether to long or short his position first. Thus she/he is not pre-committed to the strategies in Sections~\ref{ls} and \ref{sl}, and clearly this flexibility increases his overall expected profit from the trading.
The trading problem again can be formulated
 as the double optimal stopping one
 \begin{equation} \label{0-problem-1} \hs{7pc}
V^0(t,x)=\sup \limits_{0\le\tau,\zeta\le T-t}\EE\left[\e^{-r\zeta} X^x_\zeta\m \e^{-r\tau} X^x_\tau\right]
 \end{equation}
for $(t,x)\in[0,T)\times \R$ where $X^x$ represents that the process $X$ starts from $X^x_0=x$, $r>0$ is the interest rate, and
 the supremum is taken over all pairs of $\cF^{X}$- stopping times $(\tau,\zeta)$ such that $\tau,\zeta\le T-t$.
As before, $\tau$ is time for long position and $\zeta$ is time for short position.
The main difference of \eqref{0-problem-1} compare to both long-short and short-long strategies is that there is no order and constraint between $\tau$ and $\zeta$. We do not need to consider sequentially exit and entry problems, but just solve optimal problems for both long and short positions independently. Therefore we split the trading problem into the two separate problems
 \begin{align} \label{0-problem-2} \hs{7pc}
&V^{1,L}(t,x)=\sup \limits_{0\le\zeta\le T-t}\EE \left[\e^{-r\zeta} X^x_\zeta\right]\\
\label{0-problem-3}
&V^{2,L}(t,x)=\inf \limits_{0\le\tau\le T-t}\EE \left[\e^{-r\tau} X^x_\tau\right]
 \end{align}
and we have  
 \begin{equation} \label{0-problem-4} \hs{7pc}
V^0(t,x)=V^{1,L}(t,x)-V^{2,L}(t,x)
 \end{equation}
for $(t,x)\in[0,T)\times \R$ and where $V^{1,L}$ and $V^{2,L}$ are given in \eqref{th-1-1}-\eqref{th-3-1}.

Both problems \eqref{0-problem-2}-\eqref{0-problem-3} have been solved already in previous sections. The optimal entry time in \eqref{0-problem-1} is given by $\zeta_*^{1,L}\wedge\tau_*^{2,L}$ and the exit time is $\zeta_*^{1,L}\vee\tau_*^{2,L}$, where $\zeta_*^{1,L}$ and $\tau_*^{2,L}$ have been characterized   as well.  
\vs{6pt}


In Figure 2, we illustrate the two  optimal boundaries representing the long and short entering positions under the chooser strategy. We compare it to the optimal entering thresholds in the perpetual version of the problem (see Chapter 14 of \cite{HFTbook}). Intuitively, with an infinite horizon ahead, the trader can afford to wait longer and enter the market when the spread is wider in either direction. This is confirmed in Figure 2 as the optimal   boundary to long (resp. short) is above (resp. below) the optimal thresholds from the perpetual case. In other words, the continuation region, in which the trader waits to enter the market, is larger in the perptual case than in the current finite-horizon problem.

\begin{figure}[t]
\begin{center}
\includegraphics[scale=0.95]{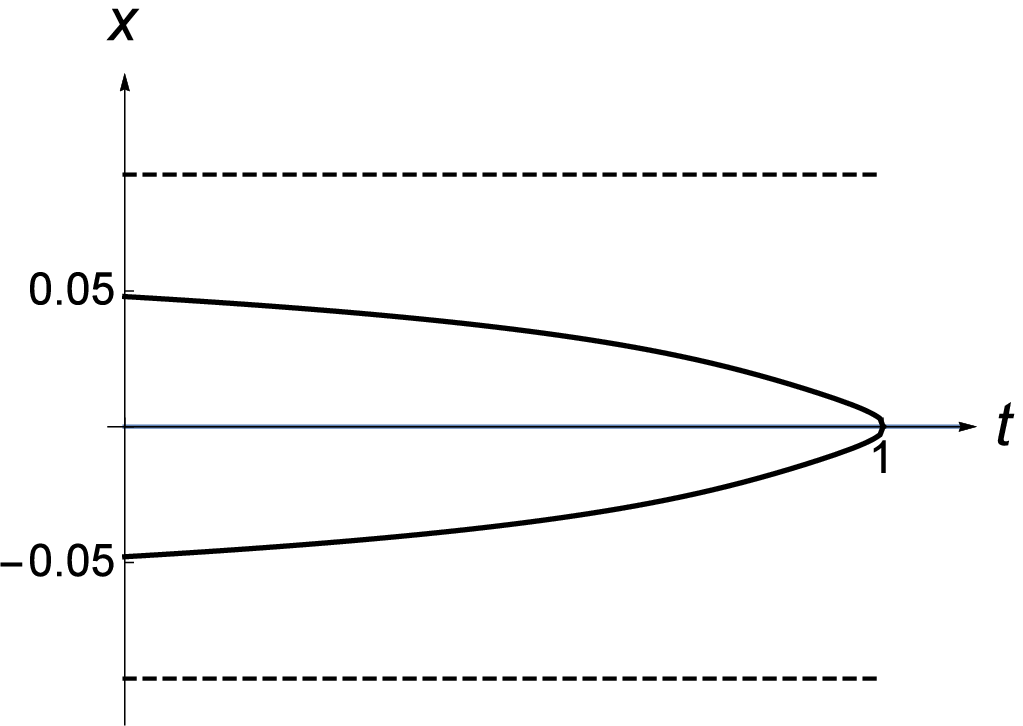}
\end{center}

{\par \leftskip=1.6cm \rightskip=1.6cm \small \ni \vs{-10pt}

\textbf{Figure 2.} The   optimal boundaries
$(b^{1,L},b^{2,L})$ (solid line) for the chooser strategy problem \eqref{0-problem-1} computed as the solution to integral equations \eqref{th-1-2} and \eqref{th-3-2}. Dashed lines represent optimal thresholds for the perpetual case.
The parameter set is $T=1$ year, $r=0.01$, $\theta=0$, $\mu=16$, $\zeta=0.16$. We used $N=500$ steps for the time discretization of interval $[0,T]$.

\par} \vs{10pt}

\end{figure}

 Here we reformulate the problem \eqref{0-problem-1} sequentially as for the long-short and short-long strategies. We already know the solution to the problem from previous
paragraph, but would like to show that the solution satisfies the free-boundary problem for the entry problem.
Once the trader enters into the position, she/he solves one of the optimal liquidation problems and both of them were already solved Sections 3.1 and 4.
Therefore we only need to study the optimal entry problem and this can be formulated as follows
 \begin{equation} \label{0-problem-5} \hs{5pc}
V^{0,E}(t,x)=\sup \limits_{0\le\tau\le T-t}\EE \left[\e^{-r\tau}G(t\p\tau,X^x_\tau)\right]
 \end{equation}
where the payoff function $G$ reads
 \begin{equation} \label{0-payoff} \hs{5pc}
G(t,x)=\max(V^{1,L}(t,x)\m x,x\m V^{2,L}(t,x))
 \end{equation}
for $t\in[0,T)$ and $x\in \R$. The payoff function $G$ shows that at entry time the trader maximizes his value and chooses the best option, i.e. whether to go long or short the spread.
Below we show that $V^{0,E}$ is the same as $V^0$ from \eqref{0-problem-4}.
   \vs{2pt}

 It can be seen that $V^{1,L}(t,x)- x=0$ for $x\ge b^{1,L}(t)$ and $x-V^{2,L}(t,x)=0$ for $x\le b^{2,L}(t)$.
 Also since $V^{1,L}$ and $V^{2,L}$ are convex and concave, respectively, we have   $V^{1,L}_x \le 1$ and $V^{2,L}_x \le1$. Hence
 the function $V^{1,L}(t,x)\m x$ is decreasing for $x< b^{1,L}(t)$ and $x\m V^{2,L}(t,x)$ is increasing $x> b^{2,L}(t)$,
 and we can conclude that there exists threshold $m(t)$ for fixed $t\in[0,T)$ such that
 \begin{equation} \label{0-payoff-2} \hs{5pc}
G(t,x)=(V^{1,L}(t,x)\m x) I(x\le m(t)) + (x\m V^{2,L}(t,x))I(x> m(t))
 \end{equation}
for $t\in[0,T)$ and $x\in \R$.
 Clearly, the function $G$ is convex in $x$ for fixed $t\in[0,T)$.
 \vs{2pt}

Given that we already know the solution to the problem  \eqref{0-problem-1}, we will use "guess-verify" method for the finite horizon optimal entry problem \eqref{0-problem-1} unlike in Sections 3.1 and 4 where the optimal boundaries were constructed directly (as solutions to the integral equations).  Let us take the pair of optimal exit strategies $(b^{1,L},b^{2,L})$ as the candidate for the optimal entry boundaries such that
the entry time is given by
\begin{align} \label{0-OST} \hs{5pc}
\tau_{0,E}=\inf\ \{\ 0\leq s\leq T\m t:X^x_{s}\le b^{2,L}(t\p s)\;\text{or}\;X^x_{s}\ge b^{1,L}(t\p s) \}
 \end{align}
 and define
  \begin{equation} \label{0-value} \hs{5pc}
\wh{V}^{0,E}(t,x)=V^{1,L}(t,x)-V^{2,L}(t,x)
 \end{equation}
for $t\in[0,T)$ and $x\in \R$ as the candidate value function for $V^{0,E}$.
\vs{2pt}

Using established properties \eqref{PDE}-\eqref{FBP2} and \eqref{PDE-3}-\eqref{FBP2-3} of the value functions ${V}^{1,L}$ and $V^{2,L}$, and boundaries $b^{1,L}$ and $b^{2,L}$,
we can verify that $\wh{V}^{0,E}$ and $(b^{1,L},b^{2,L})$ solve the following free-boundary problem
\begin{align} \label{mix-PDE} \hs{5pc}
&\wh{V}^{0,E}_t \p\L_X \wh{V}^{0,E} \m r\wh{V}^{0,E} =0 &\hs{-30pt}\text{in}\;  \cC^{0,E}\\
\label{mix-IS1}&\wh{V}^{0,E} (t,b^{1,L}(t))=G(t,b^{1,L}(t)) &\hs{-30pt}\text{for}\; t\in[0,T)\\
\label{mix-IS2}&\wh{V}^{0,E} (t,b^{2,L}(t))=G(t,b^{2,L}(t)) &\hs{-30pt}\text{for}\; t\in[0,T)\\
\label{mix-SP1}&\wh{V}^{0,E}_x (t,b^{1,L}(t))=G_x(t,b^{1,L}(t)) &\hs{-30pt}\text{for}\; t\in[0,T) \\
\label{mix-SP2}&\wh{V}^{0,E}_x (t,b^{2,L}(t))=G_x(t,b^{2,L}(t)) &\hs{-30pt}\text{for}\; t\in[0,T) \\
\label{mix-FBP1}&\wh{V}^{0,E}(t,x)>G(t,x) &\hs{-30pt}\text{in}\; \cC^{0,E}\\
\label{mix-FBP2}&\wh{V}^{0,E}(t,x)=G(t,x) &\hs{-30pt}\text{in}\; \cD^{0,E}
\end{align}
where the continuation set $\cC^{0,E}$ and the entry set $\cD^{0,E}$ are given by
\begin{align} \label{mix-C-1} \hs{5pc}
&\cC^{0,E}= \{\, (t,x)\in[0,T)\! \times\! \R:b^{2,L}(t)<x<b^{1,L}(t)\, \} \\[3pt]
 \label{mix-D-1}&\cD^{0,E}= \{\, (t,x)\in[0,T)\! \times\!\R:x\le b^{2,L}(t)\,\text{or}\,x\ge b^{1,L}(t)\, \}.
 \end{align}
 Let us show, for example, that \eqref{mix-IS1} holds indeed. This condition is equivalent to $V^{1,L}(t,b^{1,L}(t))-V^{2,L}(t,b^{1,L}(t))=b^{1,L}(t)- V^{2,L}(t,b^{1,L}(t))$ as
 $b^{1,L}(t)>m(t)$. The latter is true as $V^{1,L}(t,b^{1,L}(t))=b^{1,L}(t)$ due to \eqref{IS}. The conditions \eqref{mix-IS1}-\eqref{mix-SP2} can be shown in similar way.
\vs{2pt}

Finally, standard verification arguments indicate  that  $\wh{V}^{0,E}$ and $(b^{1,L},b^{2,L})$ are indeed the    value function  and optimal  boundaries, respectively.
 In Figure 3, we compare the value function of the chooser strategy over a finite horizon  ($T=1$ year) to  the value of the  perpetual counterpart. As we can see,  the difference is quite significant as a longer horizon allows the trader to wait longer to capture a wider spread. This also   shows the practical importance of studying     the optimal spread problem over the finite horizon. 

\begin{figure}[t]
\begin{center}
\includegraphics[scale=0.95]{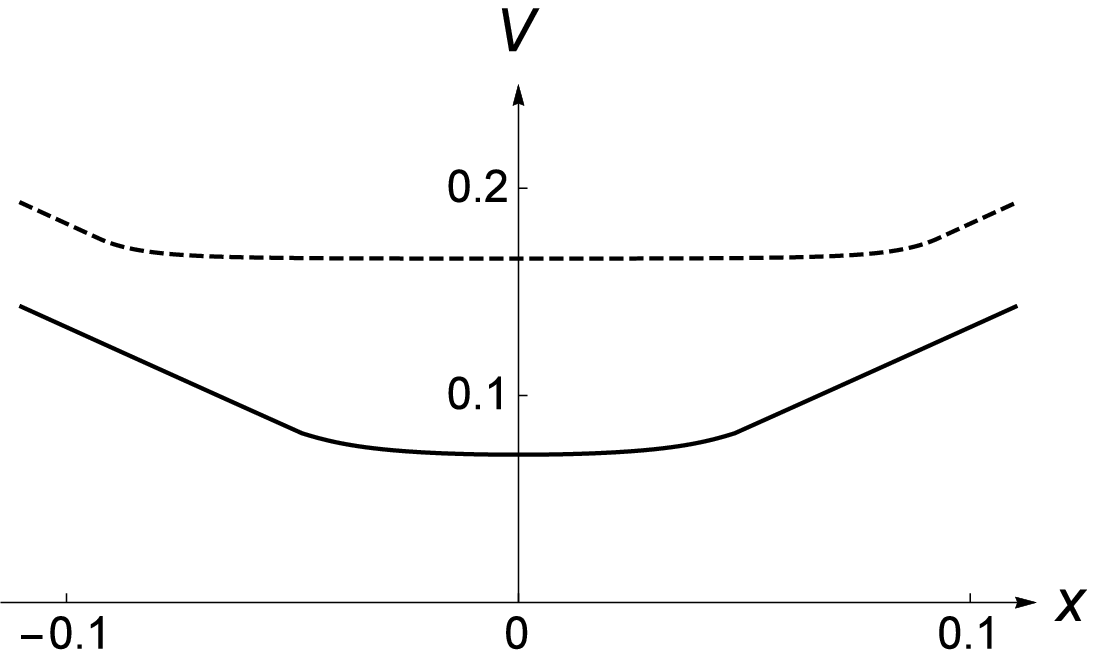}
\end{center}

{\par \leftskip=1.6cm \rightskip=1.6cm \small \ni \vs{-10pt}

\textbf{Figure 3.} The value function $V^{0,E}$ (solid line) of the chooser strategy \eqref{0-problem-1} over a finite horizon lies below  the value function for the perpetual case (dashed line). The parameters are: $T=1$ year, $r=0.01$, $\theta=0$, $\mu=16$, $\zeta=0.16.$

\par} \vs{10pt}

\end{figure}

 
  \begin{remark}
 We note that the analytical  results  above can be extended to other mean-reverting  models of   the form 
 \begin{align}\label{extension-OU}\hs{5pc}
dX_t=\mu(\theta\m X_t)dt+\sigma(X_t) dB_t, \quad X_0=x,
\end{align}
where $\sigma(x)$ is some smooth function of $X$, but not necessarily a constant. The corresponding integral equations and value function representations   will be of the same form as those under the OU model derived   above. Indeed, the  linear payoffs considered herein   render the diffusion coefficient $\sigma(x)$   irrelevant when we apply Ito's calculus. However, for numerical analysis and computation, it is crucial to know marginal distributions of $X$. For example, $X_t$ is Gaussian when $\sigma(x)=\sigma$, and $X_t$ is non-central chi-squared
when $\sigma(x)=\sigma \sqrt{x}$ (CIR process). Hence,
our   results can be extended to   models in \eqref{extension-OU}  with a known probability density  for $X_t$ that is explicit or can be approximated. We refer to \cite{LeungLiWang2014CIR} for optimal double stopping of the CIR process. 
 \end{remark}

\section{Incorporating  transaction costs}\label{trans-costs}
In this section, we incorporate    fixed transaction costs when the trader  is pre-commited to the long-short strategy. We assume that she/he may not enter into
the long position, e.g., if we enter very close to $T$ it is more likely that we end up with the loss since we pay transaction costs twice and gain at most small difference from the spread. Thus it would be optimal not to enter at all and get zero payoff.

We formulate this problem sequentially, first assume that there is open long position in the spread which want to liquidate optimally
 \begin{equation} \label{tc-problem-1} \hs{7pc}
V^{1,L,c}(t,x)=\sup \limits_{0\le\zeta\le T-t}\EE \left[\e^{-r\zeta} (X^x_\zeta-c)\right]
 \end{equation}
and the only difference with the problem \eqref{l-s-problem-2} above is that we add fixed trading fee $c>0$. Then having solved \eqref{tc-problem-1} we consider optimal entry problem
  \begin{equation} \label{tc-problem-2} \hs{4pc}
V^{1,E,c}(t,x)=\sup \limits_{0\le\tau\le T-t}\EE \left[\e^{-r\tau}(V^{1,L,c}(t\p \tau,X^x_\tau)\m X^x_\tau\m c)^+\right]
 \end{equation}
as at time $t\p\tau$ we pay $X^x_\tau+c$ and get the long position with the value $V^{1,L,c}(t\p \tau,X^x_\tau)$ and we will go long only if the payoff of this strategy is positive, otherwise
we use our right not to enter into it at this instance.

Observe that  the optimal stopping problem in \eqref{tc-problem-1} is a slight generalization of that in \eqref{l-s-problem-2}.  Therefore, to avoid repetition, we  only state the results below.
\begin{theorem}\label{th:tc}
The value function $V^{1,L,c}$ has the following representation
\begin{align}\label{th-tc-1} \hs{2pc}
V^{1,L,c}(t,x)=\e^{-r(T-t)}(m(T\m t,x)\m c) +\int_0^{T-t}\cK^{1,L,c}(u,x,b^{1,L}(t\p u))du
\end{align}
for $t\in[0,T]$ and $x\in \R$. The optimal exit boundary $b^{1,L,c}$ can be characterized as the unique solution to a nonlinear integral equation
\begin{align}\label{th-tc-2} \hs{2pc}
b^{1,L,c}(t)=\e^{-r(T-t)}(m(T\m t,b^{1,L,c}(t))\m c) +\int_0^{T-t} \cK^{1,L,c}(u,b^{1,L,c}(t),b^{1,L,c}(t\p u))du
\end{align}
for $t\in[0,T]$ in the class of continuous decreasing functions $t\mapsto b^{1,L,c}(t)$ with $b^{1,L,c}(T)=(\mu\theta\p rc)/(\mu\p r)$
where
\begin{align*} \hs{2pc}
&\cK^{1,L,c}(u,x,z)=-\e^{-ru}\EE \big[H^{1,L,c}(X^x_u) I(X^x_u \ge z)\big]\\
&H^{1,L,c}(x)=-(\mu\p r)x+\mu\theta+rc.
\end{align*}
\end{theorem}

Now  we turn to    the entry problem \eqref{tc-problem-2}. It differs from \eqref{l-s-problem-3} in two ways: there is transaction fee $c$ and, more importantly, the right not to enter into the position. Indeed, when $t$ goes $T$, the value of long position $V^{1,L,c}(t,x)$ is close to $x-c$ so that the payoff $V^{1,L,c}(t,x)-x-c$ tends to $-2c$ and thus it almost always not rational to go long near $T$. To formalize this observation, we define the curve $\gamma$ on $[0,T)$ as
\begin{equation}\hs{6pc}
V^{1,L,c}(t,\gamma(t))-\gamma(t)-c=0
\end{equation}
for $t\in[0,T)$. Hence when $x\ge \gamma(t)$ we should not enter as the value is non-positive. From the properties of $V^{1,L,c}(t,x)$, it can be seen that
$\gamma$ is decreasing  with $\gamma(T-)=-\infty$ and that $\gamma<b^{1,L,c}$.

The optionality  is the key component that  precludes us to perform the complete theoretical analysis and prove regularity properties. The problem becomes very challenging and is left for future research. Here,  we conclude the paper with  a number of   open questions with our remarks:

\begin{itemize}

\item \emph{The existence of the optimal entry boundary $b^{1,E,c}$} that separates the continuation and exercise sets. Intuitively,
it is should be true that $\cD^{1,E,c}= \{\, (t,x)\in[0,T)\! \times\!\R:x\le b^{1,E,c}(t)\, \}$.

\item \emph{Monotonicity of the boundary $b^{1,E,c}$}. Most likely, it is decreasing and explodes to $-\infty$ at $T$. For another example of this boundary behavior, we refer to    \cite{LeungLiLiZheng2015}, where the optimal futures trading problem with the transaction costs has  been numerically solved  using a finite-difference method applied to the associated  variational inequality.

\item \emph{Smooth-fit property at $b^{1,E,c}$}. The standard proof uses that the process enters immediately into the exercise region if starts
slightly above $b^{1,E,c}$. However, if the boundary is decreasing, it is not clear that this property holds.
Thus one has to compare the asymptotic behavior of the process at 0 and the slope of the boundary. The latter is unknown. In particular, the slope is very negative near $T$ and
we do not see strong evidence that the smooth fit holds close to $T$.
This open problem is general for optimal stopping problems when the immediate hitting of the boundary is not guaranteed.

\item \emph{Local time}. It is  unclear whether the local time term is present in the expression for the value function and/or in the integral equation for the optimal exercise boundary  for     problem \eqref{tc-problem-2}. The local time term will add significant challenges to  the analysis and numerical implementation of the associated integral equations.
\end{itemize}

\vs{12pt}

\begin{small}
\begin{spacing}{0.7}
\bibliographystyle{apa}
\bibliography{mybib2}
\end{spacing}
\end{small}

\end{document}